\documentclass[11pt]{article}
\usepackage[letterpaper,margin=1in]{geometry}
\usepackage{amsmath,amssymb,amsthm,mathtools}
\usepackage{aliascnt}
\usepackage{algorithm}
\usepackage{algpseudocode}
\usepackage{enumitem}
\usepackage{booktabs}
\usepackage{hyperref}
\usepackage[capitalize,nameinlink]{cleveref}
\usepackage{xcolor}
\usepackage{xspace}
\usepackage{tabularx}

\hypersetup{
  colorlinks=true,
  linkcolor=blue!55!black,
  citecolor=blue!55!black,
  urlcolor=blue!55!black
}

\newtheorem{theorem}{Theorem}[section]
\newaliascnt{lemma}{theorem}
\newtheorem{lemma}[lemma]{Lemma}
\aliascntresetthe{lemma}
\newaliascnt{corollary}{theorem}
\newtheorem{corollary}[corollary]{Corollary}
\aliascntresetthe{corollary}
\newaliascnt{proposition}{theorem}
\newtheorem{proposition}[proposition]{Proposition}
\aliascntresetthe{proposition}
\newaliascnt{definition}{theorem}
\newtheorem{definition}[definition]{Definition}
\aliascntresetthe{definition}
\newaliascnt{remark}{theorem}
\newtheorem{remark}[remark]{Remark}
\aliascntresetthe{remark}

\crefname{theorem}{Theorem}{Theorems}
\Crefname{theorem}{Theorem}{Theorems}
\crefname{lemma}{Lemma}{Lemmas}
\Crefname{lemma}{Lemma}{Lemmas}
\crefname{corollary}{Corollary}{Corollaries}
\Crefname{corollary}{Corollary}{Corollaries}
\crefname{proposition}{Proposition}{Propositions}
\Crefname{proposition}{Proposition}{Propositions}
\crefname{definition}{Definition}{Definitions}
\Crefname{definition}{Definition}{Definitions}
\crefname{remark}{Remark}{Remarks}
\Crefname{remark}{Remark}{Remarks}

\usepackage{thmtools}
\newcommand{\proofpart}[1]{
  \par\smallskip
  \noindent\textbf{#1}\enspace
}
\newcommand{\R}{\mathbb{R}}
\newcommand{\Z}{\mathbb{Z}}

\newcommand{\polylog}{\operatorname{polylog}}
\newcommand{\area}{\operatorname{area}}

\newcommand{\IGR}{\textnormal{\textsc{Min-cost Grid Matching}}\xspace}
\newcommand{\Ballp}{\mathsf{B}_p}
\newcommand{\BallOne}{\mathsf{B}_1}
\newcommand{\GridBallp}{\mathsf{B}_{p,\Z}}
\newcommand{\GridBallOne}{\mathsf{B}_{1,\Z}}

\newcommand{\normp}[1]{\left\lVert #1\right\rVert_{p}}
\newcommand{\normI}[1]{\left\lVert #1\right\rVert_{\infty}}
\newcommand{\normone}[1]{\left\lVert #1\right\rVert_{1}}
\newcommand{\calC}{\mathcal{C}}

\title{Rectilinear Matching to the Integer Grid in Nearly-Linear Time}
\author{Yu Gao}
\date{}

\begin{document}
\maketitle

\begin{abstract}
Rectilinear matching to the integer grid asks to assign each of $n$ points
in $\mathbb R^2$ to a distinct point of $\mathbb Z^2$, minimizing total
$\ell_1$ movement.  The main difficulty is that the target set is infinite:
one must first identify a finite set of relevant grid points without losing
optimality.

We prove a geometric compression theorem for this infinite-target problem.
In $O(n\log^2 n)$ time, we construct a set $\calC$ of asymptotically optimal size $O(n)$ such that, simultaneously for
every $p\in[1,\infty]$, some optimal $\ell_p$ assignment uses only points
of $\calC$.  The construction is independent of the subsequent optimization
algorithm and of the coordinate spread.

For the rectilinear case, we combine this candidate set with a linear-size
sparse network representation of $\ell_1$ distances.  In the word-RAM model with $O(1)$-word dyadic coordinates and $O(\log n)$ fractional bits, a nearly-linear time
minimum-cost flow algorithm then gives a randomized exact algorithm with
expected running time $\widetilde O(n)$.  This improves the standard
$\widetilde O(n^2)$ approach.  Combined with existing
finite geometric matching algorithms, the same candidate set also gives an
$\widetilde O(n\sqrt n\log(1/\varepsilon))$-time $(1+\varepsilon)$
approximation for every fixed integer $p\ge1$.
\end{abstract}

\section{Introduction}
We study a common core legalization task in many geometric optimization pipelines: the input is an indexed sequence $P=(h_1,\ldots,h_n)$. Each $h_i$ is a point in $\R^2$; repeated locations are allowed. The task is to find an injective map
\[
  \mu:[n]\hookrightarrow \Z^2
\]
minimizing
\[
  \sum_{i=1}^n \normone{h_i-\mu(i)}.
\]
The injectivity constraint ensures that no two objects are assigned to the same grid point, and the $\ell_1$ objective is the rectilinear movement cost.  The target set is the entire integer grid, not a finite set listed in the input.  Thus an exact algorithm must identify a relevant set of grid points and optimize the assignment to those points. Our main contribution is an explicit construction of such a set with $O(n)$ grid points.

We are not aware of a prior worst-case subquadratic exact algorithm
for the unrestricted infinite-grid problem defined above. The standard reduction that retains
$n$ nearest grid points per input, combined with the finite
geometric partial-matching algorithm of Agarwal, Chang, and
Xiao~\cite{AgarwalChangXiao2019}, gives an
$\widetilde O(n^2)$-time exact algorithm.  For each input point $h_i$, retain its $n$ nearest grid points, breaking ties arbitrarily.  Suppose an optimal assignment matches $h_i$ outside this list.  Every listed point is no farther from $h_i$, and the other $n-1$ inputs occupy at most $n-1$ of them.  At least one listed point is therefore unused, and reassigning $h_i$ to that point does not increase the cost.  Restricting the target set to the union of these nearest grid points thus preserves the optimal value.  This candidate set may, however, have size $N=\Theta(n^2)$.  Applying the exact geometric partial-matching algorithm~\cite{AgarwalChangXiao2019} to the resulting finite instance takes $\widetilde O(n^2+N)=\widetilde O(n^2)$ time.

Even when approximate solutions are allowed, the quadratic-size candidate set remains the main bottleneck.  A nearly-linear time approximation algorithm for finite geometric matching or transportation still takes $\Omega(n^2)$ time after this reduction because the reduced input can have size $\Omega(n^2)$ in the worst case.

We give a randomized exact algorithm running in expected $\widetilde O(n)$ time for rectilinear matching to the integer grid.  First, we show that a single explicitly constructible set of $O(n)$ grid points simultaneously contains, for every $\ell_p$ norm, the target grid points of at least one optimal assignment.  Second, for the $\ell_1$ norm, we optimize over these points without building the complete bipartite graph, using a sparse minimum-cost flow representation. Since an explicit assignment contains $n$ target grid
points, any algorithm that outputs the assignment requires
$\Omega(n)$ time.  Thus our expected $\widetilde O(n)$ running time
is optimal up to polylogarithmic factors.  For every fixed integer $p\ge 1$, we also obtain a fast approximation algorithm under the $\ell_p$ norm by applying existing geometric bipartite-matching algorithms to the input points and the linear-size candidate set.

In real-world applications, the target set is usually not literally the unbounded grid.  Nevertheless, such instances remain close to our setting when the legal locations come from a large grid with a list of exclusions or a bounding window.  We therefore include in Appendix~\ref{sec:endpoint-restrictions} exact solutions for these two target-restricted variants.

\subsection{Our Results}

Our main technical contribution is a general geometric reduction showing that the
infinite target grid can be compressed to linearly many candidate grid
points simultaneously for all $p \in [1, \infty]$.  The construction time does not depend on the coordinate spread or on $p$.

\begin{definition}
    We call the following problem \IGR.  The input is an indexed sequence
\[
  P=(h_1,\ldots,h_n)\in(\R^2)^n,
\]
where repeated locations are allowed, and a fixed $\ell_p$ norm with $p\in[1,\infty]$.  The task is to find an injective map
\[
  \mu:[n]\hookrightarrow \Z^2
\]
minimizing
\[
  \sum_{i=1}^n \normp{h_i-\mu(i)}.
\]
\end{definition}
We use \IGR for the general $\ell_p$ version.  The rectilinear problem in
the title is the $\ell_1$ special case, for which we obtain the randomized exact
nearly-linear time algorithm.
\begin{theorem}[Universal optimal-size candidate set for $\ell_p$ norms]\label{thm:general-candidate-set}
There is an algorithm that explicitly constructs a set $\calC\subseteq \Z^2$ with the following properties.
\begin{enumerate}[label=(\roman*)]
  \item For every $p\in[1,\infty]$, there is an optimal assignment $\mu_p^*$ for \IGR under $\|\cdot\|_p$ with $\mu_p^*(i)\in \calC$ for every $i\in[n]$.
  \item The size of $\calC$ satisfies
  \[
    |\calC|=O(n).
  \]
  \item The algorithm runs in time
  \[
    O(n\log^2 n).
  \]
\end{enumerate}
\end{theorem}
\begin{remark}[Tightness of the candidate-set bound]
\label{rem:candidate-set-optimality}
Every candidate set contains the target points of some injective
assignment of the $n$ inputs.  These target points are distinct,
and therefore every candidate set has cardinality at least $n$.
Consequently, the bound $|\calC|=O(n)$ in
\cref{thm:general-candidate-set} is asymptotically optimal.
The construction time $O(n\log^2 n)$ is likewise optimal up to
polylogarithmic factors in the explicit-output model.
\end{remark}

For the $\ell_1$ norm, this leads to an exact nearly-linear time algorithm.

\begin{restatable}[Rectilinear matching to the integer grid in nearly-linear time]{theorem}{LoneMainAlgorithm}\label{thm:main-algorithm}
In the word-RAM model with word size $w=\Theta(\log n)$,
suppose that there is a global parameter $B=O(\log n)$
such that every input coordinate belongs to $2^{-B}\mathbb Z$
and its exact scaled representation fits in $O(1)$ machine
words. Then there is a randomized exact algorithm that outputs an
optimal assignment in expected $\widetilde O(n)$ time.
\end{restatable}

We may also pass the candidate set from \cref{thm:general-candidate-set} to the approximation algorithm for finite geometric matching~\cite{AgarwalChangXiao2019}.  This immediately gives approximation algorithms for \IGR for $\ell_p$ norms with integer $p\ge 1$.

\begin{restatable}[Approximation for positive-integer $\ell_p$ norms]{corollary}{LpAlgorithm}\label{thm:lp-algorithm}
Let $p\ge1$ be a fixed integer and let $0<\varepsilon\le1$. There is an $\widetilde O(n\sqrt{n}\log(1/\varepsilon))$-time algorithm
that outputs a solution for \IGR under $\normp{\cdot}$ whose cost is at most $(1+\varepsilon)$ times the optimum.
\end{restatable}

For the $\ell_1$ norm, the remaining ingredient of the exact algorithm is a simple sparse optimization graph.  It avoids the dense bipartite graph between input points and candidate grid points by representing $\ell_1$ distances as paths in a shared rectilinear network.
\begin{restatable}[Sparse $\ell_1$ flow reduction]{theorem}{LoneFlowReduction}\label{thm:l1-flow-reduction}
For the $\ell_1$ norm, given the input points and the candidate set of \cref{thm:general-candidate-set}, we construct a directed minimum-cost flow instance with
$O(n)$
vertices and edges.  The underlying undirected graph belongs to a $1/2$-separable graph family~\cite{DongGaoGoranciLeePengSachdevaYe2025} with separator construction time $s(m)=\widetilde O(m)$, and the instance has the same optimal value as \IGR under the $\ell_1$ norm.
\end{restatable}

We then invoke a nearly-linear time minimum-cost flow algorithm for separable graph families as a black box; its standard guarantee assumes polynomially bounded integral costs~\cite{DongGaoGoranciLeePengSachdevaYe2025}.

We note that the pruning algorithm for constructing \cref{thm:general-candidate-set} is simple to implement.  It uses only standard sorting and orthogonal range-counting data structures.  In real-world $\ell_1$ legalization instances, it can therefore serve as a lightweight front end: it reduces the infinite-grid assignment to a sparse graph that can be passed to existing minimum-cost flow solvers.  We also record two common target-restricted variants in Appendix~\ref{sec:endpoint-restrictions}.

\subsection{Technical Overview}\label{sec:overview}
We next outline the proof.
\paragraph{Safe radius vector.}
The safe radius vector $\rho$ is the main invariant: for every $p$, there exists an optimal assignment under $\|\cdot\|_p$ that assigns each input $h_j$ to a grid point within distance $\rho_j$.  Initially, we set $\rho_i=O(\sqrt n)$ because an $\ell_p$ ball of radius $\rho_i$ contains $\Theta(\rho_i^2)$ grid points, with constants uniform over $p\in[1,\infty]$, whereas the other input points can occupy at most $n-1$ of them.  The final construction uses the $p$-independent enclosing-square candidate set
\[
  \left(\bigcup_i
  \left(h_i+[-\rho_i,\rho_i]^2\right)\right)\cap \Z^2.
\]
The formal definition and initialization are given in \cref{sec:preliminaries,sec:safe-pruning}.

\paragraph{Safe pruning.}
Suppose the current radius vector is safe.  For a point $h_i$ and a test radius $r$, consider the grid points in $\Ballp(h_i,r)$.  If $h_i$ were matched farther than $r$, all those grid points would have to be used by other input points, which we call the supporters of $h_i$.  A supporter $h_j$ can use one of these grid points only if $h_j$ is close to $h_i$ at scale $r+\rho_j$.  We upper-bound the number of potential supporters using a two-dimensional range-counting computation.  If the number of grid points exceeds this upper bound, then the same witness assignment must match $h_i$ within distance $r$, and we may safely update $\rho_i$ to $r$; see \cref{sec:safe-pruning}.

As a warm-up, applying this pruning rule with $r=n^{1/4}$ and $\rho_i=n^{1/2}$ gives a candidate set of size $O(n\sqrt{n})$; see \cref{sec:warmup}.

\paragraph{Iterated pruning.}
To further reduce the number of candidates, we apply the pruning rule with geometrically decreasing test radii until the radius is constant.  Let $R_t$ denote the test radius in pruning round $t$.  If the radius of $h_i$ is updated to $R_t$, then $h_i$ contributes at most $O(R_t^2)$ candidates.  Otherwise, the square-covering estimator is large, so \(h_i\) has
many potential supporters. Either some potential supporter already has large radius, in which case $h_i$ lies near the current charged region, or all potential supporters have small radius, in which case $h_i$ lies in a new dense region at scale $O(R_{t-1})$.  We charge the latter regions to input points not charged in previous rounds; the charged sets are disjoint, so the total high-radius area is linear.  The full iterated construction and area bound are in \cref{sec:iterated-pruning}, and they give \cref{thm:general-candidate-set}.

\paragraph{The $\ell_1$ flow reduction.}
For the $\ell_1$ norm, the distance from a real input point to an integer grid point decomposes through the boundary of the unit grid cell containing the input point, and a shortest path in the grid graph.  This reduction is proved in \cref{sec:l1-flow} and gives the flow instance of \cref{thm:l1-flow-reduction}.

\paragraph{Target restrictions.}
The same framework also handles two common restrictions on target grid points.  A finite set of unavailable grid points is handled by adding dummy points and applying an exchange argument that places every dummy at its designated forbidden site.  For a sufficiently thick rectangular window, the lower bounds on grid-point counts used by pruning still hold inside the window, and the $\ell_1$ flow graph is restricted to available candidates.  These variants are proved in Appendix~\ref{sec:endpoint-restrictions}.

\subsection{Related Work}

\paragraph{Applications.}

Many geometric optimization pipelines first compute preferred locations in
the plane and only later enforce a discrete feasibility constraint.  A direct
example arises in VLSI physical design: in hybrid-bonding and related 3D-IC
flows, bonding terminals, bumps, or pads must be assigned to available
manufacturing grid sites, with no site used more than once, while keeping
their displacement from routed or planned locations small
~\cite{ShiGaoRenXieXuXueYuanQianZhou2025,HuangPentapatiAgnesinaBrunionLim2024,
VannaIampikulYoonParkYeapLim2025}.
Existing legalization approaches use force-directed and other heuristics, or
formulate a terminal-to-site assignment problem and invoke a general-purpose
matching or optimization solver.  Such approaches can be effective in
practice, but they either sacrifice exactness or do not exploit the grid
geometry to obtain a worst-case nearly-linear running time.  Related
assignment abstractions arise in standard-cell legalization, where cells are
moved from continuous placement locations to nonoverlapping legal positions
in rows while limiting displacement
~\cite{HougardyNeuwohnerSchorr2021}, and in grid-map visualization, where
representatives of geographic regions are assigned one-to-one to grid cells
while minimizing Manhattan or related distances
~\cite{EppsteinKreveldSpeckmannStaals2015}.  Motivated by this common core,
we study \IGR under $\ell_1$ and other $\ell_p$ norms.

\paragraph{Finite Geometric Matching and Transportation.}
Finite geometric matching assumes that both sides of the matching instance are explicit point sets.  Agarwal, Chang, and Xiao~\cite{AgarwalChangXiao2019} give an exact $O((N+k^2)\polylog N)$-time algorithm for planar geometric partial matching with finite input size $N$ and matching size $k$.  After the nearest-grid reduction, their algorithm takes $\widetilde O(n^2)$ time for rectilinear \IGR.

For approximation, Agarwal et al. and Fox and Lu~\cite{AgarwalChangRaghvendraXiao2022,FoxLu2020} give $O(n\log^{O(d)}n\,\varepsilon^{-O(d)})$-time algorithms for finite geometric matching and transportation in $\R^d$.  The planar approximation algorithm of Agarwal, Chang, and Xiao~\cite{AgarwalChangXiao2019} computes a $(1+\varepsilon)$-approximate geometric matching of size $k$ in $\widetilde O((n+k\sqrt{k})\log(1/\varepsilon))$ time.  Combined with the $O(n)$-size candidate set from \cref{thm:general-candidate-set}, it gives an $\widetilde O(n\sqrt{n}\log(1/\varepsilon))$-time approximation algorithm for \IGR under $\normp{\cdot}$ for integer $p\ge 1$.  Sharathkumar and Agarwal~\cite{SharathkumarAgarwal2012} solve $\ell_1$ and $\ell_\infty$ geometric bipartite matching in $O(n\sqrt{n}\log^{d+O(1)}n\log\Delta)$ time in $\R^d$ when the input points have diameter $\Delta$.

Under the $\ell_1$ norm, finite geometric partial matching 
also admits a generic sparse minimum-cost flow formulation.   Given point sets
$A$ and $B$, four two-dimensional dominance range trees produce a directed
network of size
\[
  O\bigl(|B|\log |B|+|A|\log^2 |B|\bigr)
\]
that represents every pair $(a,b)\in A\times B$ by a path of cost
$\|a-b\|_1$.  When the coordinates can be scaled to polynomially bounded
integers, the almost-linear-time minimum-cost flow algorithm of
Chen et al.~\cite{ChenKyngLiuPengGutenbergSachdeva2022} therefore gives an
exact $(|A|+|B|)^{1+o(1)}$-time algorithm.  Combined with
\cref{thm:general-candidate-set}, this yields an alternative randomized exact
$n^{1+o(1)}$-time algorithm for \IGR under polynomially bounded coordinate
diameter.  Unlike our main $\ell_1$ reduction, however, this generic construction
is not linear-size and does not expose the separable structure needed for our
$\widetilde O(n)$ bound. Since this construction is not needed for the main result, we give the construction and proof in
Appendix~\ref{sec:range-tree-flow}.
\paragraph{Grid Rounding and Snapping.}
Snap rounding and grid snapping also move geometric objects to a grid, but their objectives are different.  Snap rounding converts segment arrangements into fixed-precision representations while maintaining robust geometric structure~\cite{GoodrichGuibasHershbergerTanenbaum1997,GuibasMarimont1998,Hobby1999}, and iterated snap rounding strengthens separation guarantees~\cite{HalperinPacker2002}.  Grid snapping has also been studied for graph drawings; L\"offler, van Dijk, and Wolff~\cite{LoefflerDijkWolff2016} show NP-hardness for several objectives when the embedding of a planar graph must be preserved.  These problems involve arrangements or graph-topological constraints rather than a min-cost injective assignment of independent points to grid sites.

\paragraph{Planar Minimum-cost Flow.}
Dong et al.~\cite{DongGaoGoranciLeePengSachdevaYe2025} give a nearly-linear time algorithm for minimum-cost flow in planar graphs with polynomially bounded integer costs and capacities, and their framework extends to separable graph families.  We use their algorithm after the $\ell_1$ reduction has produced an $O(n)$-edge separable instance; our candidate-set theorem is independent of this flow solver.

\section{Preliminaries}\label{sec:preliminaries}
Throughout the paper, subscripts on asymptotic notation indicate which quantities the hidden constants may depend on; for example, $O_D(\cdot)$ may depend on $D$. We use $\widetilde O(\cdot)$ to hide polylogarithmic factors.

\paragraph{Computational Model.}
We work in the word-RAM model with word size
$w=\Theta(\log n)$. There is a global nonnegative precision parameter
$B=O(\log n)$ such that every input coordinate belongs to
$2^{-B}\mathbb Z$ and its scaled integer representation fits in
$O(1)$ machine words. Coordinates are manipulated as exact dyadic
rationals; no rounded floating-point arithmetic is used.

Arithmetic and comparisons on $O(1)$-word integers, bit shifts, and
floor and ceiling operations on dyadic rationals take constant time.
All intermediate integers generated by the algorithm fit in $O(1)$
machine words.
\paragraph{Coordinate Normalization.}
The assumption that the integer part of every coordinate is representable
using $O(1)$ machine words does not restrict the combinatorial problem, up
to a preprocessing step whose running time is linear in the input encoding. Intuitively, we may partition the input points into components by connecting pairs of input points within $O(n)$ distance of each other. Each component can be processed separately with $O(\log n)$-bit integer part per input coordinate. We defer the details to \cref{sec:coordinate-normalization}.

\subsection{Norms and Balls}

Fix $p\in[1,\infty]$, and let $\|\cdot\|_p$ be the corresponding $\ell_p$ norm on $\R^2$.  We write
\[
  \Ballp(h,r)=\{x\in\R^2: \normp{x-h}\le r\}.
\]
We also write
\[
  \GridBallp(h,r)=\Ballp(h,r)\cap \Z^2.
\]

For a set $Y\subseteq\R^2$ and a radius $r\ge 0$, we write
\[
  \mathsf N_r(Y)
  =
  \{x\in\R^2:\operatorname{dist}_\infty(x,Y)\le r\}
\]
for the $\ell_\infty$-neighborhood of $Y$. $\mathsf N_r(\emptyset)$ is defined as $\emptyset$.

\begin{lemma}[Uniform norm comparison]\label{lem:norm-comparison}
For every $p\in[1,\infty]$ and every $x\in\R^2$,
\[
  \normI{x}\le \normp{x}\le 2\normI{x}.
\]
The factor $2$ is the smallest constant that works uniformly over all such $p$.
\end{lemma}
\begin{proof}
For every $p\in[1,\infty]$,
\[
  \normI{x}\le \normp{x}\le \normone{x}\le 2\normI{x}.
\]
For $p=1$ and $x=(1,1)$, the last inequality is an equality, proving
uniform optimality of the factor $2$.
\end{proof}

\begin{lemma}[Grid points in norm balls]\label{lem:lattice-growth}
For every $p\in[1,\infty]$, every $h\in\R^2$, and every $r\ge2$,
\[
  \frac{r^2}{4} \le |\GridBallp(h,r)| \le 9r^2.
\]
\end{lemma}
\begin{proof}
By \cref{lem:norm-comparison},
\[
  h+[-r/2,r/2]^2\subseteq \Ballp(h,r)
  \subseteq h+[-r,r]^2.
\]
Each coordinate interval of the inner square has length $r$ and hence
contains at least $\lfloor r\rfloor\ge r/2$ integers.  Each coordinate
interval of the outer square has length $2r$ and contains at most
$2r+1\le3r$ integers.  Squaring these two bounds proves the claim.
\end{proof}

\begin{lemma}[$\ell_\infty$ packing bound]\label{lem:linf-packing}
Let $X\subset\R^2$ be a set whose points are pairwise more than $s>0$ apart in $\ell_\infty$.  Then every $\ell_\infty$-ball of radius $R$ contains at most $O((1+R/s)^2)$ points of $X$.
\end{lemma}
\begin{proof}
Place an $\ell_\infty$-ball of radius $s/2$ around each point of $X$.  These balls are interior-disjoint.  If their centers lie in an $\ell_\infty$-ball of radius $R$, then the balls themselves lie in the concentric $\ell_\infty$-ball of radius $R+s/2$.  Comparing areas gives $O((1+R/s)^2)$ centers.
\end{proof}

\subsection{Safe Radius Vectors}
As the target set $\Z^2$ is infinite, we need to formally prove the existence of an optimum assignment.
\begin{lemma}[Existence of an optimum]
Every instance of \IGR has an optimal assignment.
\end{lemma}

\begin{proof}
Choose any feasible injective assignment $\mu_0$ and let
\[
  U=\sum_{i=1}^n \|h_i-\mu_0(i)\|_p.
\]
Any assignment of cost at most $U$ assigns each $h_i$ to a grid
point in $\Ballp(h_i,U)$.  Each such ball contains only finitely
many grid points.  Hence there are only finitely many injective
assignments of cost at most $U$, and the minimum among them is
attained.
\end{proof}
\begin{definition}[Safe radius vector]\label{def:safe-radius}
A vector $\rho=(\rho_1,\ldots,\rho_n)\in \R^n_{\ge 0}$ is safe for $P$ under $\|\cdot\|_p$ if there exists an optimal assignment $\mu^*$ such that
\[
  \normp{h_i-\mu^*(i)}\le \rho_i
  \qquad\text{for every }i\in[n].
\]
Such an assignment is called a witness for $\rho$.
\end{definition}

\begin{definition}[Candidate set]\label{def:candidate-set}
A set $\calC\subseteq\Z^2$ is a candidate set for \IGR if there exists an optimal assignment $\mu^*$ such that $\mu^*(i)\in\calC$ for every $i\in[n]$.
\end{definition}

\section{Safe Pruning}\label{sec:safe-pruning}

This section proves the basic pruning rule used by both the warm-up and the iterated-pruning algorithm.  We first show that the algorithm can start from a uniform safe radius.  We then introduce the square-covering estimator, which is the efficiently computable overcount used in the pruning test.  Finally, we prove that every successful pruning step preserves safety.

\subsection{Initialization}

The pruning procedure needs an initial safe vector.  The following lemma shows that a global radius of order $\sqrt n$ always suffices: if an optimal assignment sends one point farther than this, then all closer grid points must already be occupied.

\begin{lemma}[Initial safe radius]\label{lem:global-safe-radius}
For every $p\in[1,\infty]$, the vector $\rho_i=3\sqrt n$ is safe under $\|\cdot\|_p$.
\end{lemma}
\begin{proof}
Fix $p$ and let $\mu$ be any optimal assignment.  Suppose that
$\normp{h_i-\mu(i)}=d$.  If $d\le3$, then $d\le3\sqrt n$.  Otherwise,
every grid point in $\GridBallp(h_i,d-1)$ must be used by $\mu$: if such
a point were unused, reassigning $h_i$ to it would preserve injectivity
and strictly reduce the cost.  Hence $|\GridBallp(h_i,d-1)|\le n$.
Since $d-1>2$, \cref{lem:lattice-growth} gives
\[
  \frac{(d-1)^2}{4}\le n.
\]
Thus $d\le2\sqrt n+1\le3\sqrt n$.  This holds for every $i$, so $\mu$
witnesses the claimed safe vector.  As $p$ was arbitrary, the same radius
bound is valid uniformly over $p\in[1,\infty]$.
\end{proof}

\subsection{The Square-Covering Estimator}

The pruning rule intends to compare the number of grid points near $h_i$ with the number of other input points that could possibly use those grid points. We call such input points potential supporters of $h_i$ at radius $r$.  Computing the number of potential supporters exactly for an arbitrary $\ell_p$ norm is unnecessary.  Instead, we use a simple overcount based on axis-aligned squares.  This is the only place where the algorithm uses a geometric data structure.

Let $\rho$ be a radius vector and let $r\ge 0$.  For every input point $h_j$, define the $\ell_\infty$-square
\[
  Q_j^\rho(r)=h_j+[-(r+\rho_j),r+\rho_j]^2.
\]
For each input point $h_i$, define the square-covering estimator
\[
  \widehat M_i^\rho(r)
  =
  \left|\{j\in[n]\setminus\{i\}:h_i\in Q_j^\rho(r)\}\right|.
\]

\begin{lemma}[Square-covering overcount]\label{lem:estimator-overcount}
Assume that $z\in\GridBallp(h_i,r)$ and $\normp{h_j-z}\le \rho_j$.  Then $h_i\in Q_j^\rho(r)$.  Therefore $\widehat M_i^\rho(r)$ is an upper bound on the number of potential supporters of $h_i$ under the radius vector $\rho$.
\end{lemma}
\begin{proof}
The triangle inequality gives
\[
  \normp{h_i-h_j}\le \normp{h_i-z}+\normp{z-h_j}\le r+\rho_j.
\]
Since $\normI{x}\le \normp{x}$, we have
\[
  \normI{h_i-h_j}\le r+\rho_j.
\]
Equivalently, $h_i\in Q_j^\rho(r)$.  Thus every potential supporter is counted by $\widehat M_i^\rho(r)$; the estimator may count additional indices, which is safe for pruning.
\end{proof}

\subsection{The Pruning Rule}

With this overcount in place, we can prove the pruning rule.  The key point is that a single optimal assignment witnesses safety.  If there are enough grid points near $h_i$ for the other input points to occupy, then the same witness must already assign $h_i$ nearby.

\begin{lemma}[Safe pruning step]\label{lem:safe-pruning}
Let $\rho$ be a safe radius vector under $\|\cdot\|_p$.  Let $r\ge2$.  If
\[
  \frac{r^2}{4} > \widehat M_i^\rho(r),
\]
then the vector $\rho'$ defined by
\[
  \rho'_i=r,
  \qquad
  \rho'_j=\rho_j \quad\text{for }j\ne i
\]
is safe.
\end{lemma}
\begin{proof}
Let $\mu$ be an optimal assignment witnessing $\rho$.  If $\normp{h_i-\mu(i)}\le r$, then $\mu$ already witnesses $\rho'$.  Otherwise, every grid point in $\GridBallp(h_i,r)$ must be used by an input point other than $h_i$; if one were unused, assigning $h_i$ to it would strictly reduce the cost of $\mu$.

For each $z\in\GridBallp(h_i,r)$, let $h_j$ be the input point assigned to $z$ by $\mu$.  Then $j\ne i$ and $\normp{h_j-z}\le \rho_j$, because $\mu$ witnesses the safety of $\rho$.  By \cref{lem:estimator-overcount}, this index $j$ is counted by $\widehat M_i^\rho(r)$.  Since $\mu$ is injective, distinct grid points $z$ give distinct indices $j$.  Therefore
\[
  |\GridBallp(h_i,r)|\le \widehat M_i^\rho(r).
\]
This contradicts \cref{lem:lattice-growth} and the assumption
$r^2/4>\widehat M_i^\rho(r)$.  Hence $\normp{h_i-\mu(i)}\le r$, and the same optimal assignment witnesses $\rho'$.
\end{proof}

Because the preceding proof never changes the witness assignment, all successful local updates can be applied at once.  This simultaneous version is what the algorithm uses in each round.

\begin{lemma}[Simultaneous safe pruning]\label{lem:simultaneous-safe-pruning}
Let $\rho$ be a safe radius vector.  Let $r\ge2$.  Let $S\subseteq[n]$ be the set of indices satisfying
\[
  \frac{r^2}{4} > \widehat M_i^\rho(r).
\]
Define
\[
  \rho'_i=\begin{cases}
    r, & i\in S,\\
    \rho_i, & i\notin S.
  \end{cases}
\]
Then $\rho'$ is safe.
\end{lemma}
\begin{proof}
Let $\mu$ be an optimal assignment witnessing $\rho$.  The proof of \cref{lem:safe-pruning} shows that, for every $i\in S$, this same assignment must satisfy $\normp{h_i-\mu(i)}\le r$; otherwise, we would obtain a contradiction to the pruning inequality for that particular $i$.  For indices outside $S$, the old bound $\rho_i$ is unchanged.  Thus $\mu$ witnesses the safety of $\rho'$.
\end{proof}

\subsection{Computing the Estimator}

It remains to evaluate the overcount efficiently.  For a fixed radius, this is the standard offline problem of counting how many axis-aligned squares contain each query point.

\begin{lemma}[Square-covering computation]\label{lem:square-covering}
For fixed $r$ and $\rho$, all values $\widehat M_i^\rho(r)$ can be computed in $O(n\log n)$ time by counting, for each input point, how many of the squares $Q_j^\rho(r)$ cover it.
\end{lemma}
\begin{proof}
Create two vertical sweep events for each square $Q_j^\rho(r)$, one on each of its two vertical sides.  Sort these events together with the query points $h_i$ by $x$-coordinate.  At a common $x$-coordinate, process left-side insertions first, then queries, and right-side deletions last, so that boundary containment is counted inclusively. During the sweep, maintain the $y$-intervals of the active squares in a Fenwick tree after coordinate-compressing all relevant $y$-coordinates.  When the sweep reaches $h_i$, the data structure returns the number of active squares whose $y$-interval contains the $y$-coordinate of $h_i$.  Finally, subtract the contribution of $Q_i^\rho(r)$, which always covers $h_i$.  Sorting and Fenwick-tree operations take $O(n\log n)$ time.
\end{proof}

\section{A Warm-Up for \texorpdfstring{$\ell_1$}{L1}: \texorpdfstring{$O(n\sqrt n)$}{O(n sqrt n)} Candidates}\label{sec:warmup}

As a warm-up, we show that in the $\ell_1$ norm, a single global reduction already turns the quadratic candidate set into an $O(n\sqrt n)$ candidate set.  The $\ell_1$ case keeps the geometry transparent while exposing the two ingredients used later for general norms: low-radius balls can be handled directly at low cost, and high-radius balls can occur only in dense parts of the input.
\begin{definition}[One-round pruning radii]\label{def:one-round-radius}
Set $R_{-1}=2^{\lceil \log_2 (3\sqrt n)\rceil}$.  Starting from
$\rho_i^{(-1)}=R_{-1}$, define $R_0=2^{\lceil\log_2(3n^{1/4})\rceil}$ and apply
\cref{lem:simultaneous-safe-pruning} with test radius $R_0$ under the
$\ell_1$ norm.  The resulting vector is denoted by $\rho^{(0)}$.
\end{definition}

The candidate set consists of the grid points in the final safe $\ell_1$ balls.  Points whose radius was reduced to $R_0$ contribute only $O(R_0^2)=O(\sqrt{n})$ grid points each.  The following theorem shows that the union of the balls centered at the remaining points, whose radii stay large, has small total area.

\begin{theorem}[Warm-up candidate bound]\label{thm:warmup}
The candidate set
\[
  \calC_0=\bigcup_{i=1}^n \GridBallOne(h_i,\rho_i^{(0)})
\]
is a candidate set for the $\ell_1$ norm and satisfies
\[
  |\calC_0|=O(n\sqrt n).
\]
It can be constructed in $O(n\sqrt n\log n)$ time.
\end{theorem}
\begin{proof}
By \cref{lem:simultaneous-safe-pruning}, the radius vector $\rho^{(0)}$ is safe, meaning that $\calC_0$ contains the grid points used by some optimal assignment.

After the pruning round, $\rho_i^{(0)}\in \{R_{-1}, R_0\}$. Let $H_0=\{i:\rho_i^{(0)}=R_{-1}\}$.
We split the candidate region into the low-radius part and the high-radius part.  If $i\notin H_0$, then $\rho_i^{(0)}= R_0$, and the total number of grid points contributed by these balls is at most
\[
  \sum_{i\notin H_0} O((R_0+1)^2)=O(nR_0^2+n)=O(n\sqrt n).
\]
For $i\in H_0$, we will see that the contribution is also $O(n\sqrt n)$ by \cref{lem:warmup-packing}.  This proves the size bound.  The construction time follows from one call to \cref{lem:square-covering} and output-sensitive enumeration of the resulting grid points.
\end{proof}

It remains to bound the number of grid points in the high-radius region.  An index whose radius remains $R_{-1}$ must fail the pruning test at radius $R_0$ and therefore has many potential supporters within distance $O(\sqrt n)$.  A maximal separated subfamily lets us bound the number of these dense regions.

\begin{lemma}[Warm-up high-radius packing]\label{lem:warmup-packing}
Let $H_0=\{i:\rho_i^{(0)}>R_0\}$.  Then
\[
  \left|\left(\bigcup_{i\in H_0}\BallOne(h_i,\rho_i^{(0)})\right)\cap \Z^2\right|=O(n\sqrt n).
\]
\end{lemma}
\begin{proof}
Every $i\in H_0$ failed the pruning test at radius $R_0$, so
\[
  \widehat M_i^{\rho^{(-1)}}(R_0)\ge R_0^2/4=\Omega(\sqrt n).
\]
Because $\rho_j^{(-1)}=R_{-1}=2^{\lceil \log_2 (3\sqrt n)\rceil}$ for all $j$, every point counted by this estimator lies at $\ell_\infty$ distance $O(\sqrt n)$ from $h_i$.

Choose a maximal subset $J\subseteq H_0$ whose points $\{h_i:i\in J\}$ are pairwise more than $R_{-1}$ apart in $\ell_\infty$.  By maximality, every center $h_i$ with $i\in H_0$ is at $\ell_\infty$-distance at most $R_{-1}$ from some center $h_j$ with $j\in J$.  Hence the high-radius region $\bigcup_{i\in H_0}\BallOne(h_i,\rho_i^{(0)})$ is covered by $\ell_\infty$-squares of radius $O(\sqrt n)$ centered at the points indexed by $J$.  These squares contain $O(|J|n)$ grid points in total.

It remains to bound $|J|$.  Each $i\in J$ has $\Omega(\sqrt n)$ input points at $\ell_\infty$-distance $O(\sqrt n)$ from $h_i$.  Conversely, by \cref{lem:linf-packing}, any fixed input point lies in only $O(1)$ of these neighborhoods because the centers in $J$ are $R_{-1}$-separated.  Counting center--input incidences gives $|J|\Omega(\sqrt n)\le O(n)$, so $|J|=O(\sqrt n)$.  This gives the required bound.
\end{proof}

\section{Optimal-Size Candidates by Iterated Pruning}\label{sec:iterated-pruning}

To obtain a linear-size candidate set, we repeat the same safe pruning rule with a fixed geometric shrink factor until the test radius becomes constant, as shown in \cref{alg:iterated-pruning}. The algorithm is independent of $p$, and all constants are uniform over $p\in[1,\infty]$. The resulting radius vector is the same for all $p$. For the analysis
of the high-radius regions (\cref{sec:high-radius-region}, \cref{sec:high-radius-charging}), we will fix an arbitrary $p\in[1,\infty]$.

\subsection{The Test Radii Sequence}
\label{sec:radii-seq}
Let
\[
  R_{-1}=2^{\lceil \log_2 (3\sqrt n)\rceil}
\]
and for $t\ge 0$ let
\[
  R_t=R_{t-1}/2.
\]
Let $T\ge0$ be the largest integer with $R_T\ge2$.  Since $R_{-1}$ is
a power of two and $R_{-1}\ge4$, this index always exists,
$T=O(\log n)$, and $R_T=2$.
Since $R_{-1} \ge 3\sqrt{n}$, the vector $\rho^{(-1)}$ is safe for every $p\in[1,\infty]$ by \cref{lem:global-safe-radius}.
The algorithm starts from $\rho_i^{(-1)}=R_{-1}$ for all $i$.  For each $t=0,1,\ldots,T$, it computes all values
\[
  \widehat M_i^{\rho^{(t-1)}}(R_t)
\]
and applies \cref{lem:simultaneous-safe-pruning} at radius $R_t$.  The resulting radius vector is $\rho^{(t)}$.

\begin{algorithm}[t]
\caption{Iterated pruning}\label{alg:iterated-pruning}
\begin{algorithmic}[1]
\Require Indexed input points $P=\{h_1,\ldots,h_n\}$
\State Set $R_{-1}\gets 2^{\lceil \log_2 (3\sqrt n)\rceil}$, and choose $T$ as in \cref{sec:radii-seq}
\For{$i=1,\ldots,n$}
  \State $\rho_i^{(-1)}\gets R_{-1}$
\EndFor
\For{$t=0,1,\ldots,T$}
  \State $R_t\gets R_{t-1}/2$
  \State Compute all values $\widehat M_i^{\rho^{(t-1)}}(R_t)$ using the square-covering estimator (\cref{lem:square-covering})
  \For{$i=1,\ldots,n$}
    \If{$R_t^2/4>\widehat M_i^{\rho^{(t-1)}}(R_t)$}
      \State $\rho_i^{(t)}\gets R_t$
    \Else
      \State $\rho_i^{(t)}\gets \rho_i^{(t-1)}$
    \EndIf
  \EndFor
\EndFor
\State \Return $\rho^{(T)}$
\end{algorithmic}
\end{algorithm}

The choice of test radii makes $R_t^2$ shrink by a factor of $4$ from one round to the next.  At the end, the low-radius points have radius exactly $R_T=2$ and contribute only $O(n)$ candidates.  The remaining points are accounted for by an amortized area bound.

\subsection{The High-Radius Region}
\label{sec:high-radius-region}
\begin{definition}[High-radius indices and region]\label{def:high-radius-set}
For $t=-1,0,1,\ldots,T$, define
\[
  H_t=\{i\in[n]:\rho_i^{(t)}>R_t\}.
\]
The high-radius region is
\[
  U_t=\bigcup_{i\in H_t}\Ballp(h_i,\rho_i^{(t)}).
\]
\end{definition}

Since a radius changes only when it is reset to the current test radius, after round $t$ every radius is one of $R_{-1},R_0,\ldots,R_t$.  Thus, for $t\ge 0$, if $i\notin H_t$ then $\rho_i^{(t)}=R_t$, while if $i\in H_t$ then $\rho_i^{(t)}$ is one of $R_{-1},R_0,\ldots,R_{t-1}$.  In particular, if $i\in H_t\setminus H_{t-1}$, then
\[
  \rho_i^{(t-1)}=\rho_i^{(t)}=R_{t-1}.
\]
Although iterated pruning may resemble a binary search for a displacement bound for each input, its outcomes are not monotone across rounds. The test in round $t$ uses the current radius vector $\rho^{(t-1)}$ which changes as other inputs are pruned. Consequently, an input may fail the test at a larger radius in one round but pass at a smaller radius in a later round. Thus, the high-radius sets $H_t$ need not be nested. The analysis does not rely on monotonicity of $H_t$; it uses only the monotonicity of the charged regions $S_t$ as defined in \cref{sec:high-radius-charging}.
For a constant $D\ge 1$, let
\[
  U_t^{(D)}=\bigcup_{i\in H_t}\Ballp(h_i,D\rho_i^{(t)}).
\]

The following packing lemma generalizes \cref{lem:warmup-packing}.

\begin{lemma}[Dense-region packing bound]\label{lem:dense-region-packing-lattice}
Let $R\ge 1$, $\eta>0$, and $D\ge 1$ be fixed.  Let $X\subseteq[n]$, and
let $S\subseteq[n]$ be a set of centers such that every $i\in S$ satisfies
\[
  \left|\{j\in X:\|h_i-h_j\|_\infty\le R\}\right|
  \ge \eta R^2.
\]
Then
\[
  \area\left(\bigcup_{i\in S}
    \bigl(h_i+[-DR,DR]^2\bigr)
  \right)
  =
  O_D\left(\frac{|X|}{\eta}\right),
\]
and also
\[
  \left|
  \left(
  \bigcup_{i\in S}
    \bigl(h_i+[-DR,DR]^2\bigr)
  \right)
  \cap\Z^2
  \right|
  =
  O_D\left(\frac{|X|}{\eta}\right).
\]
\end{lemma}

\begin{proof}
Take a maximal subset $J\subseteq S$ whose centers are pairwise more than
$R$ apart in $\ell_\infty$.  By maximality, the squares centered at the points
indexed by $S$ with radius $DR$ are covered by those centered at the points
indexed by $J$ with radius
$(D+1)R$.

For each $j\in J$, there are at least $\eta R^2$ points of $X$ at
$\ell_\infty$-distance at most $R$ from $h_j$.  Conversely, because the centers in
$J$ are $R$-separated, \cref{lem:linf-packing} implies that every fixed
input point belongs to only $O(1)$ of these $R$-neighborhoods.  Therefore
\[
  |J|\eta R^2\le O(|X|).
\]
The area covered by the enlarged squares is $O_D(|J|R^2)$, which gives
the area bound.  Since each enlarged square contains $O_D(R^2)$ integer
grid points, the same inequality also gives the lattice-point-count bound.
\end{proof}

The next lemma is the local dichotomy used by the charging argument.  If a point $h_i$ cannot be pruned to radius $R_t$, then the square-covering estimator identifies many input indices near $h_i$. Either one of these potential supporters already belongs to the previous high-radius region, or all potential supporters had radius exactly $R_{t-1}$ in the previous round and lie near $h_i$.  In the latter case, we add a dense region near $h_i$ and charge its grid points to the potential supporters.

\begin{lemma}[Pruning-failure dichotomy]\label{lem:failed-pruning-dichotomy}
Set $H_{-1}=\emptyset$.  Fix $t\in\{0,\ldots,T\}$ and let $i\in H_t\setminus H_{t-1}$.  Let
\[
  T_i=\{j\ne i:h_i\in Q_j^{\rho^{(t-1)}}(R_t)\}.
\]
Then $|T_i|\ge R_{t-1}^2/16$, and one of the following holds.
\begin{enumerate}[label=(\roman*)]
  \item $T_i$ contains an index in $H_{t-1}$.
  \item $T_i\cap H_{t-1}=\emptyset$, and every $j\in T_i$ satisfies
  \[
    \normI{h_i-h_j}<2R_{t-1}.
  \]
\end{enumerate}
\end{lemma}
\begin{proof}
Since $i\in H_t$, it was not pruned at radius $R_t$.  Because $i\in H_t\setminus H_{t-1}$, $\rho_i^{(t-1)}=\rho_i^{(t)}=R_{t-1}$.  The failed pruning test gives
\[
  \widehat M_i^{\rho^{(t-1)}}(R_t)\ge R_t^2/4
  = R_{t-1}^2/16.
\]
  This is exactly the lower bound on $|T_i|$.

If $T_i$ contains an index in $H_{t-1}$, the first outcome holds.  Otherwise, every $j\in T_i$ has $\rho_j^{(t-1)}=R_{t-1}$.  Since $h_i\in Q_j^{\rho^{(t-1)}}(R_t)$,
\[
  \normI{h_i-h_j}\le R_t+R_{t-1}<2R_{t-1},
\]
which proves the second outcome.
\end{proof}
\subsection{The High-Radius Charging Argument}
\label{sec:high-radius-charging}
We now prove the amortized bound on the final high-radius region.

Every scale used below is bounded below by the absolute constant $2$.
We fix $b=12$.
Recall that $R_t=R_{t-1}/2$.

\paragraph{Charged regions for the analysis.}
We construct charged regions $S_t$, for $t=-1,0,\ldots,T$, satisfying
\[
  U_t^{(2)}\subseteq S_t.
\]
The auxiliary notation follows a simple progression.  The set $I_t$
identifies the centers that require new dense regions, $X_t$ contains the potential
supporters charged for those regions, $V_t$ is the corresponding geometric
region, and $S_t$ accumulates all regions charged up to round $t$. All these objects are introduced only for the analysis and are
not maintained by the pruning algorithm.
\medskip
\begin{center}
\small
\begin{tabular}{@{}c p{0.73\linewidth}@{}}
\toprule
Symbol & Role in the charging argument \\
\midrule
$\widetilde S_t,\ S_t$
  & The expansion of the previously charged region and the cumulative
    charged region after round $t$, respectively. \\

$I_t$
  & The newly high-radius centers that require genuinely new dense
    regions. \\

$X_t,\ V_t$
  & The new potential supporter indices charged in round $t$ and the corresponding
    geometric region. \\

$\Delta_t$
  & The total amount by which $V_t$ is expanded in all later rounds. \\
\bottomrule
\end{tabular}
\end{center}
\medskip

We next define these objects formally and establish, first, that $S_t$
covers the high-radius region, and second, that the new region $V_t$ can be
charged to the fresh potential supporter set $X_t$.

Set $S_{-1}=\emptyset$.  For every round $t=0,1,\ldots,T$, define
\[
  \widetilde S_t=\mathsf N_{bR_t}(S_{t-1}).
\]
For every index $i$ that enters $H_t$ as the threshold decreases in round $t$, recall
\[
  T_i=\{j\ne i:h_i\in Q_j^{\rho^{(t-1)}}(R_t)\}.
\]
We define the set of indices that create genuinely new dense regions:
\[
  I_t =
  \left\{
  i\in H_t\setminus H_{t-1}:
  T_i\cap H_{t-1}=\emptyset
  \text{ and }
  \Ballp(h_i,2\rho_i^{(t)})\nsubseteq \widetilde S_t
  \right\}.
\]
Let
\[
  X_t=\bigcup_{i\in I_t}T_i
\]
be the set of input-point indices charged in round $t$, and let
\[
  V_t
  =
  \bigcup_{i\in I_t}
  \left(h_i+[-4R_{t-1},4R_{t-1}]^2\right).
\]
Finally, set
\[
  S_t=\widetilde S_t\cup V_t.
\]

\begin{lemma}[Coverage invariant]\label{lem:charged-region-covers-high-balls}
For every $t=0,1,\ldots,T$,
\[
  U_t^{(2)}\subseteq S_t.
\]
\end{lemma}

\begin{proof}
We prove the claim by induction on $t$.  The case $t=-1$ is vacuous.  Assume
that $U_{t-1}^{(2)}\subseteq S_{t-1}$.

Fix $i\in H_t$.  If $i\in H_{t-1}$, then $i$ remains a high-radius index from the previous
round and its radius has not changed:
\[
  \rho_i^{(t)}=\rho_i^{(t-1)}.
\]
Hence
\[
  \Ballp(h_i,2\rho_i^{(t)})
  =
  \Ballp(h_i,2\rho_i^{(t-1)})
  \subseteq
  U_{t-1}^{(2)}
  \subseteq
  S_{t-1}
  \subseteq
  \widetilde S_t
  \subseteq
  S_t.
\]

It remains to consider $i\in H_t\setminus H_{t-1}$.  Then
\[
  \rho_i^{(t)}=\rho_i^{(t-1)}=R_{t-1}.
\]

First, suppose that $T_i\cap H_{t-1}\ne\emptyset$.  Choose
$j\in T_i\cap H_{t-1}$.  Since $j\in T_i$,
\[
  \|h_i-h_j\|_\infty\le R_t+\rho_j^{(t-1)}.
\]
By norm comparison,
\[
  \|h_i-h_j\|_p
  \le
  2(R_t+\rho_j^{(t-1)}).
\]
For any
\[
  x\in \Ballp(h_i,2\rho_i^{(t)})
  =
  \Ballp(h_i,2R_{t-1}),
\]
we therefore have
\[
  \|x-h_j\|_p
  \le
  2R_{t-1}
  +
  2R_t
  +
  2\rho_j^{(t-1)}.
\]
Since
\[
  \Ballp(h_j,2\rho_j^{(t-1)})
  \subseteq
  U_{t-1}^{(2)}
  \subseteq
  S_{t-1},
\]
the $\ell_p$-distance from $x$ to $S_{t-1}$ is at most
\[
  2R_{t-1}+2R_t
  =
  6R_t
  \le
  bR_t.
\]
The $\ell_\infty$-distance is no larger than the $\ell_p$-distance.  Thus
\[
  x\in \mathsf N_{bR_t}(S_{t-1})=\widetilde S_t.
\]
So the entire ball
\[
  \Ballp(h_i,2\rho_i^{(t)})
\]
is contained in $\widetilde S_t$.

Now suppose that $T_i\cap H_{t-1}=\emptyset$.  If
\[
  \Ballp(h_i,2\rho_i^{(t)})\subseteq \widetilde S_t,
\]
then there is nothing to prove.  Otherwise, $i$ is in $I_t$ by definition.
Since $\rho_i^{(t)}=R_{t-1}$,
\[
  \Ballp(h_i,2\rho_i^{(t)})
  \subseteq
  h_i+[-2R_{t-1},2R_{t-1}]^2
  \subseteq
  h_i+[-4R_{t-1},4R_{t-1}]^2
  \subseteq
  V_t
  \subseteq
  S_t.
\]
This proves the induction step.
\end{proof}
Next, we charge the area of the new region $V_t$ in round $t$ to the input points in the dense regions.  The region $V_t$ is expanded by $bR_{t'}$ in every subsequent round $t'$, but because $R_{t'}$ decreases geometrically with $t'$, these expansions do not change the asymptotic area bound.
\begin{lemma}[Cost of round $t$]\label{lem:one-round-charged-region-cost}
For every $t=0,1,\ldots,T$, define
\[
  \Delta_t=\sum_{s=t+1}^T bR_s,
\]
with $\Delta_T=0$.  Then
\[
  \area\bigl(\mathsf N_{\Delta_t}(V_t)\bigr)
  =
  O(|X_t|)
\]
and
\[
  \left|
  \mathsf N_{\Delta_t}(V_t)\cap\Z^2
  \right|
  =
  O(|X_t|).
\]
\end{lemma}

\begin{proof}
Fix a round $t$.  If $I_t=\emptyset$, the claim is trivial.  Otherwise,
take any $i\in I_t$.  By definition of $I_t$, we have
\[
  T_i\cap H_{t-1}=\emptyset.
\]
Therefore the second outcome of \cref{lem:failed-pruning-dichotomy} holds:
every $j\in T_i$ satisfies
\[
  \|h_i-h_j\|_\infty<2R_{t-1}.
\]
Moreover, because $i\in H_t\setminus H_{t-1}$ failed the pruning test at
radius $R_t$,
\[
  |T_i|
  \ge
  R_t^2/4
  =
  R_{t-1}^2/16.
\]

Set
\[
  R=2R_{t-1}
  \qquad\text{and}\qquad
  \eta=\frac{1}{64}.
\]
Then
\[
  |T_i|
  \ge
  \eta R^2,
\]
and all points of $T_i$ lie at $\ell_\infty$-distance at most $R$ from $h_i$.
Since $T_i\subseteq X_t$, the hypotheses of
\cref{lem:dense-region-packing-lattice} hold with $S=I_t$ and $X=X_t$.

The set $V_t$ is exactly
\[
  V_t
  =
  \bigcup_{i\in I_t}
  \left(h_i+[-2R,2R]^2\right),
\]
because $R=2R_{t-1}$.  Hence
\cref{lem:dense-region-packing-lattice}, applied with $D=2$, gives
\[
  \area(V_t)=O(|X_t|)
\]
and
\[
  |V_t\cap\Z^2|=O(|X_t|).
\]

It remains to account for the subsequent expansions of $V_t$.  Since
$R_s=R_{s-1}/2$,
\[
  \Delta_t
  =
  \sum_{s=t+1}^T bR_s
  \le
  \sum_{s=t+1}^\infty bR_s
  =
  bR_t
  =
  \frac{b}{2}R_{t-1}
  =
  \frac{b}{4}R.
\]
Therefore $\mathsf N_{\Delta_t}(V_t)$ is contained in the union of squares
with the same centers and half-side length
\[
  \left(2+\frac b4\right)R=5R.
\]
Applying \cref{lem:dense-region-packing-lattice} again, with
$D=5$, gives
\[
  \area\bigl(\mathsf N_{\Delta_t}(V_t)\bigr)
  =
  O(|X_t|)
\]
and
\[
  \left|
  \mathsf N_{\Delta_t}(V_t)\cap\Z^2
  \right|
  =
  O(|X_t|).
\]
\end{proof}
The next lemma shows that only input points in $S_t\setminus S_{t-1}$ can be charged in round $t$.  Its corollary, \cref{cor:charged-sets-disjoint}, shows that no input point is charged twice.
\begin{lemma}[Freshness of charged potential supporters]\label{lem:charged-supporters-are-fresh}
For every round $t$ and every index $j\in X_t$,
\[
  h_j\in S_t\setminus S_{t-1}.
\]
\end{lemma}

\begin{proof}
Fix $j\in X_t$.  Then there exists $i\in I_t$ such that $j\in T_i$.
Since $i\in I_t$, we are in the second case of
\cref{lem:failed-pruning-dichotomy}.  Therefore
\[
  \|h_i-h_j\|_\infty<2R_{t-1}.
\]
Thus the point $h_j$ lies in the square centered at $h_i$ in the
definition of $V_t$:
\[
  h_j\in
  h_i+[-4R_{t-1},4R_{t-1}]^2
  \subseteq
  V_t
  \subseteq
  S_t.
\]

We now prove that $h_j\notin S_{t-1}$.  Suppose, for contradiction, that
$h_j\in S_{t-1}$.  For every
\[
  x\in \Ballp(h_i,2\rho_i^{(t)})
  =
  \Ballp(h_i,2R_{t-1}),
\]
we have
\[
  \|x-h_i\|_\infty
  \le
  \|x-h_i\|_p
  \le
  2R_{t-1}.
\]
Together with $\|h_i-h_j\|_\infty<2R_{t-1}$, this gives
\[
  \operatorname{dist}_\infty(x,S_{t-1})
  \le
  \|x-h_j\|_\infty
  \le
  4R_{t-1}
  =8R_t
  \le bR_t.
\]
Thus every such $x$ belongs to
\[
  \mathsf N_{bR_t}(S_{t-1})=\widetilde S_t.
\]
Hence
\[
  \Ballp(h_i,2\rho_i^{(t)})
  \subseteq
  \widetilde S_t,
\]
contradicting the definition of $i\in I_t$.  Therefore
$h_j\notin S_{t-1}$, and the claim follows.
\end{proof}

\begin{corollary}[Disjointness of charged sets]\label{cor:charged-sets-disjoint}
The sets $X_0,X_1,\ldots,X_T$ are pairwise disjoint.  Consequently,
\[
  \sum_{t=0}^T |X_t|\le n.
\]
\end{corollary}

\begin{proof}
The regions $S_t$ are monotone:
\[
  S_{-1}\subseteq S_0\subseteq S_1\subseteq\cdots\subseteq S_T.
\]
By \cref{lem:charged-supporters-are-fresh}, if $j\in X_t$, then
$h_j\in S_t$.  For any later round $s>t$, this implies
$h_j\in S_{s-1}$.  But again by
\cref{lem:charged-supporters-are-fresh}, every index in $X_s$ must have
its point outside $S_{s-1}$.  Hence $j\notin X_s$.  Therefore the sets
$X_t$ are pairwise disjoint, and their total size is at most $n$.
\end{proof}
We are now ready to bound the area of the high-radius region $U_T^{(2)}$.
\begin{lemma}[High-radius area bound]\label{lem:high-radius-area}
The final high-radius region satisfies
\[
  \area\left(U_T^{(2)}\right)=O(n)
\]
and
\[
  \left|U_T^{(2)}\cap\Z^2\right|=O(n).
\]
\end{lemma}

\begin{proof}
By \cref{lem:charged-region-covers-high-balls},
\[
  U_T^{(2)}\subseteq S_T.
\]
It remains to bound the area of $S_T$ and the number of lattice points it contains.

By construction,
\[
  S_t
  =
  \mathsf N_{bR_t}(S_{t-1})\cup V_t.
\]
Unrolling this recurrence and using monotonicity of $\ell_\infty$-neighborhoods,
we get
\[
  S_T
  \subseteq
  \bigcup_{t=0}^T
  \mathsf N_{\Delta_t}(V_t),
  \qquad
  \Delta_t=\sum_{s=t+1}^T bR_s.
\]
Therefore, by \cref{lem:one-round-charged-region-cost},
\[
  \area(S_T)
  \le
  \sum_{t=0}^T
  \area\bigl(\mathsf N_{\Delta_t}(V_t)\bigr)
  =
  O\left(\sum_{t=0}^T |X_t|\right).
\]
Similarly,
\[
  |S_T\cap\Z^2|
  \le
  \sum_{t=0}^T
  \left|
  \mathsf N_{\Delta_t}(V_t)\cap\Z^2
  \right|
  =
  O\left(\sum_{t=0}^T |X_t|\right).
\]
By \cref{cor:charged-sets-disjoint},
\[
  \sum_{t=0}^T |X_t|\le n.
\]
Thus
\[
  \area(S_T)=O(n)
  \qquad\text{and}\qquad
  |S_T\cap\Z^2|=O(n).
\]
Since $U_T^{(2)}\subseteq S_T$, the same bounds hold for
$U_T^{(2)}$.
\end{proof}

\subsection{The Final Candidate Set}

The low- and high-radius area bounds complete the geometric part of the size analysis.  It remains only to replace the $\ell_p$-balls by the $\ell_\infty$-squares that contain them, which allows us to construct the candidate set explicitly.

\begin{definition}[Final candidate set]\label{def:grid-candidate-set}
The grid candidate set after the final round is
\[
  \calC=
  \bigcup_{i=1}^n
  \left(
    \left(h_i+[-\rho_i^{(T)},\rho_i^{(T)}]^2\right)\cap \Z^2
  \right).
\]
\end{definition}

\begin{lemma}[Candidate-set safety]\label{lem:endpoint-safety}
For every $p\in[1,\infty]$, the radius vector $\rho^{(T)}$ is safe under
$\|\cdot\|_p$.  Consequently, $\calC$ contains the grid points used by
an optimal assignment under each such norm.
\end{lemma}
\begin{proof}
Fix $p\in[1,\infty]$.  The initial radius vector is safe under
$\|\cdot\|_p$ by \cref{lem:global-safe-radius}.  Each round applies
\cref{lem:simultaneous-safe-pruning}, so every radius vector
$\rho^{(t)}$ is safe under this norm.  In particular, $\rho^{(T)}$ is
witnessed by an optimal assignment $\mu_p^*$ satisfying
$\mu_p^*(i)\in\GridBallp(h_i,\rho_i^{(T)})$ for all $i$.  By
\cref{lem:norm-comparison}, each such grid point lies in the
$\ell_\infty$-square used in \cref{def:grid-candidate-set}.  The
construction is independent of $p$, so this conclusion holds for every
$p\in[1,\infty]$ for the same set $\calC$.
\end{proof}

\begin{lemma}[Candidate-set size]\label{lem:candidate-size}
\[
  |\calC|=O(n).
\]
\end{lemma}

\begin{proof}
Fix any $p\in[1,\infty]$ for the following geometric bound.  Partition
the contributions to $\calC$ according to whether $i\in H_T$.  If
$i\notin H_T$, then
\[
  \rho_i^{(T)}=R_T=2,
\]
so the total number of grid points contributed by these low-radius squares
is $O(n)$.

Now consider $i\in H_T$.  The square used in the definition of $\calC$ is
contained in the inflated $p$-ball:
\[
  \left(h_i+[-\rho_i^{(T)},\rho_i^{(T)}]^2\right)
  \subseteq
  \Ballp(h_i,2\rho_i^{(T)}),
\]
because $\|x\|_p\le 2\|x\|_\infty$.  Therefore the high-radius part
of $\calC$ is contained in
\[
  U_T^{(2)}\cap\Z^2.
\]
By \cref{lem:high-radius-area}, this set has size $O(n)$.  Hence
$|\calC|=O(n)$.
\end{proof}

\begin{lemma}[Candidate-set construction time]
\label{lem:candidate-construction}
The set $\calC$ can be constructed explicitly in time
\[
  O(n\log^2 n).
\]
\end{lemma}

\begin{proof}
There are $T+1=O(\log n)$ pruning rounds.  Each round computes all
values $\widehat M_i^{\rho^{(t-1)}}(R_t)$ in $O(n\log n)$ time by
\cref{lem:square-covering} and updates the radii in linear time.

It remains to enumerate the lattice points in the union of the final
squares.  Replace the square centered at $h_i=(x_i,y_i)$ with radius
$\rho_i^{(T)}$ by the integer rectangle
\[
  \begin{aligned}
  I_i^x &=
  [\lceil x_i-\rho_i^{(T)}\rceil,
   \lfloor x_i+\rho_i^{(T)}\rfloor]\cap\Z,\\
  I_i^y &=
  [\lceil y_i-\rho_i^{(T)}\rceil,
   \lfloor y_i+\rho_i^{(T)}\rfloor]\cap\Z.
  \end{aligned}
\]
This replacement preserves exactly the lattice points of the square.

Sweep the integer $x$-columns.  A rectangle inserts its interval $I_i^y$
at the first column of $I_i^x$ and deletes it immediately after the last
column.  Maintain the union of the active $y$-intervals using a standard
dynamic interval-union structure~\cite{ChengJanardan1991}.  On each
nonempty column, report the covered integer $y$-coordinates; when the
active union is empty, jump directly to the next event.  Every processed
nonempty column outputs at least one point of $\calC$, so the sweep takes
\[
  O(n\log n+|\calC|)
  =
  O(n\log n)
\]
time by \cref{lem:candidate-size}.  Together with the pruning rounds, the
total running time is $O(n\log^2 n)$.
\end{proof}

\begin{proof}[Proof of \cref{thm:general-candidate-set}]
The algorithm is the iterated-pruning procedure described above, followed by the construction of $\calC$.  Correctness follows from \cref{lem:endpoint-safety}, the size bound from \cref{lem:candidate-size}, and the running-time bound from \cref{lem:candidate-construction}.
\end{proof}
\section{The \texorpdfstring{$\ell_1$}{L1} Flow Reduction}
\label{sec:l1-flow}

Let $\rho=\rho^{(T)}$ and let $\calC$ be the candidate set constructed
in \cref{sec:iterated-pruning}.  Thus there is an optimal assignment
$\mu^*$ satisfying
\[
  \normone{h_i-\mu^*(i)}\le \rho_i
  \qquad\text{and}\qquad
  \mu^*(i)\in\calC
\]
for every $i\in[n]$.

For a point $h=(x,y)\in\R^2$, define its two vertical anchors by
\[
  \mathsf A(h)
  =
  \{(\lfloor x\rfloor,y),(\lceil x\rceil,y)\},
\]
with duplicates removed.  For an anchor $a=(k,y)$, where $k\in\Z$,
define
\[
  \mathsf N_{\calC}(a)
  =
  \{(k,\lfloor y\rfloor),(k,\lceil y\rceil)\}\cap\calC,
\]
again with duplicates removed.

\paragraph{Network construction.}
Consider the following rectilinear segments:
\begin{enumerate}[label=(\roman*)]
  \item for every pair of grid-adjacent candidates
  $q,q'\in\calC$ with $\normone{q-q'}=1$, the unit segment
  $\overline{qq'}$;
  \item for every input point $h_i$ and every
  $a\in\mathsf A(h_i)$, the horizontal segment $\overline{h_i a}$;
  \item for every $a\in\mathsf A(h_i)$ and every
  $z\in\mathsf N_{\calC}(a)$, the vertical segment $\overline{az}$.
\end{enumerate}
Take the geometric union of these segments, merge collinear overlaps,
and subdivide at every source, anchor, candidate point, and segment
intersection.  Degenerate segments are discarded.  For every resulting
elementary segment $\overline{uv}$, add the two directed arcs
\[
  u\to v
  \qquad\text{and}\qquad
  v\to u,
\]
each with capacity $n$ and cost $\normone{u-v}$.

Finally, add a sink vertex $\tau$ and, for every $q\in\calC$, an arc
\[
  q\to\tau
\]
of capacity $1$ and cost $0$.  The vertex at $h_i$ has supply $1$, and
$\tau$ has demand $n$.  Coincident geometric vertices are identified;
in particular, supplies of coincident input points are added.

\LoneFlowReduction*
\begin{proof}
\proofpart{Size and construction time.}
The grid-induced part contains at most $2|\calC|$ undirected unit
segments.  The source-to-anchor construction contributes at most $2n$
segments, and the anchor-to-grid construction contributes at most $4n$
segments.

All noncollinear intersections occur at endpoints of source or anchor
segments.  Collinear overlaps are merged before subdivision.  Hence the
resulting rectilinear arrangement has
\[
  O(n+|\calC|)=O(n)
\]
vertices and elementary segments.  The two orientations and the
$|\calC|$ sink arcs preserve this bound.  The graph can be constructed
in $O(n\log n)$ time by sorting the segments on each supporting line
and by sorting or hashing the candidate grid points.

\proofpart{Equality of optimum values.}
Let $\mu^*$ be an optimal assignment witnessing the safety of $\rho$,
and write
\[
  q_i=\mu^*(i).
\]
Choose the anchor
\[
  a_i=
  \begin{cases}
    (\lceil (h_i)_x\rceil,(h_i)_y),
      & (q_i)_x\ge (h_i)_x,\\
    (\lfloor (h_i)_x\rfloor,(h_i)_y),
      & (q_i)_x<(h_i)_x.
  \end{cases}
\]
Next choose
\[
  z_i=
  \begin{cases}
    ((a_i)_x,\lceil (h_i)_y\rceil),
      & (q_i)_y\ge (h_i)_y,\\
    ((a_i)_x,\lfloor (h_i)_y\rfloor),
      & (q_i)_y<(h_i)_y.
  \end{cases}
\]
Then
\[
  \normone{h_i-q_i}
  =
  \normone{h_i-a_i}
  +
  \normone{a_i-z_i}
  +
  \normone{z_i-q_i}.
\]
Moreover, $z_i$ and every grid vertex on a monotone grid path from
$z_i$ to $q_i$ satisfy
\[
  \normone{h_i-z}
  \le
  \normone{h_i-q_i}
  \le
  \rho_i.
\]
Consequently, all these grid vertices belong to $\calC$ by
\cref{def:grid-candidate-set}.  The network therefore contains a path
from $h_i$ to $q_i$ of cost exactly $\normone{h_i-q_i}$.  Sending one
unit along these paths and then through $q_i\to\tau$ gives a feasible
flow, because the points $q_i$ are distinct.  Thus the minimum flow
cost is at most the optimum matching cost.

Conversely, integral capacities and demands imply the existence of an
integral optimum flow.  Remove directed flow cycles and decompose the
remaining flow into unit source-to-sink paths.  Each path enters $\tau$
through an arc $q\to\tau$.  The unit capacities of these arcs imply
that the resulting grid points are distinct.  The part of the path
from its source $h_i$ to $q$ is a rectilinear walk of cost at least
\[
  \normone{h_i-q}.
\]
Hence the paths define an injective assignment whose cost is at most
the flow cost.  The two optimum values are therefore equal.

\proofpart{Separability.}
Let $\mathfrak G$ be the subgraph closure of the underlying undirected
graphs. Since deleting $\tau$ from any graph in $\mathfrak G$ leaves a planar graph, by Remark 20 and Corollary 26 of~\cite{DongGaoGoranciLeePengSachdevaYe2025}, $\mathfrak G$ is a
subgraph-closed $1/2$-separable family with $s(m)=\widetilde O(m)$.
\end{proof}

We now apply the separable min-cost-flow algorithm of
Dong et al.~\cite{DongGaoGoranciLeePengSachdevaYe2025}.

\LoneMainAlgorithm*
\begin{proof}
Construct the candidate set $\calC$ by
\cref{thm:general-candidate-set}, and construct the flow network $G$ as
above.  These steps take $O(n\log^2 n)$ and $O(n\log n)$ time,
respectively.  By \cref{thm:l1-flow-reduction}, the network has
$m=O(n)$ arcs, belongs to a subgraph-closed $1/2$-separable family with
$s(m)=\widetilde O(m)$, and has the same optimum value as \IGR.

Let
\[
  D=2^B.
\]
Every geometric arc has length at most $1$, and all of its endpoint
coordinates belong to $2^{-B}\Z$.  After multiplying all costs by $D$,
the arc costs are integers in $\{0,\ldots,D\}$.  Capacities and demands
have absolute value at most $n$.  Therefore, with
\[
  M=\max\{n,D\}=n^{O(1)},
\]
the separable min-cost-flow algorithm of
Dong et al.~\cite{DongGaoGoranciLeePengSachdevaYe2025}, with
$\alpha=1/2$, runs in expected time
\[
  \widetilde O\left(
    (m+m^{1/2+\alpha})\log M+s(m)
  \right)
  =
  \widetilde O(n).
\]

The integral-recovery step returns an integral optimum flow $f$.
Every arc not entering $\tau$ is a nondegenerate geometric arc and
therefore has strictly positive cost; moreover, $\tau$ has no outgoing
arcs.  Hence the positive-flow support of $f$ is acyclic.  Indeed, if it
contained a directed cycle, subtracting the minimum flow value on that
cycle would preserve feasibility and strictly decrease the cost.

Process this support in topological order.  At each vertex $v$, maintain
an implicit sequence $L_v$ of source labels.  Initially, $L_v$ contains
all indices $i$ such that $h_i=v$; sequences arriving on incoming arcs
are concatenated to $L_v$.  For every outgoing arc $e=(v,w)$, split off
exactly $f(e)$ labels from $L_v$ and concatenate them to $L_w$.  Flow
conservation guarantees that the required number of labels is available,
and integrality guarantees that $f(e)$ is an integer.  Whenever
$f(q,\tau)=1$, the unique label sent through $q\to\tau$ identifies the
input assigned to $q$.

An implicit treap supports each split and concatenation in $O(\log n)$
time.  Since every arc is processed once, the total postprocessing time
is $O(m\log n+n)=\widetilde O(n)$.
\end{proof}

\section{Conclusion}

We gave a randomized exact nearly-linear time algorithm for
rectilinear matching to the integer grid, which is optimal up to
polylogarithmic factors.

The core component is a single explicitly constructible set of
$O(n)$ grid points that, for every $p\in[1,\infty]$, contains the
targets of some optimal $\ell_p$ assignment.  This gives an
asymptotically optimal finite representation of \IGR.
Consequently, exact or
approximate algorithms for finite geometric partial matching can
be transferred to grid matching without the quadratic candidate
blowup of the standard nearest-grid reduction.  The rectilinear
part of the paper then exploits additional $\ell_1$ structure to
obtain a linear-size separable flow instance.

Natural questions include whether similarly optimal candidate
sets can be obtained in higher dimensions, for other periodic
target sets, or under additional capacity and obstacle
constraints.

\section{Acknowledgments}
The author thanks Jingbang Chen for various helpful discussions. The author used language-model tools to assist with editing and
presentation.  The author independently verified the mathematical
arguments, algorithmic claims, and bibliographic claims, and assumes
responsibility for all content.

\appendix

\section{A Generic Range-Tree Flow Reduction for Finite
\texorpdfstring{$\ell_1$}{L1} Matching}
\label{sec:range-tree-flow}

This appendix records the generic range-tree reduction mentioned in the
related-work discussion.  The reduction is not used in the proof of our main
$\widetilde O(n)$ result.  Its purpose is to show that the linear candidate
set alone already gives an exact $n^{1+o(1)}$-time algorithm under a
polynomial coordinate-diameter assumption, even without exploiting the
local grid structure used in \cref{sec:l1-flow}.

Let
\[
  A=(a_1,\ldots,a_k)\in(\R^2)^k
  \qquad\text{and}\qquad
  B=\{b_1,\ldots,b_m\}\subseteq\R^2,
\]
where $k\le m$.  The finite geometric partial-matching problem asks for an
injective map
\[
  \mu:[k]\hookrightarrow B
\]
minimizing
\[
  \sum_{i\in [k]}\|a_i-\mu(i)\|_1.
\]

For a sign vector
\[
  \sigma=(\sigma_x,\sigma_y)\in\{-1,+1\}^2,
\]
define
\[
  \phi_\sigma(z)
  =
  \sigma_x z_x+\sigma_y z_y.
\]
We write $b\preceq_\sigma a$ when
\[
  \sigma_x b_x\le \sigma_x a_x
  \qquad\text{and}\qquad
  \sigma_y b_y\le \sigma_y a_y.
\]
Whenever $b\preceq_\sigma a$, the two coordinate differences have the
directions specified by $\sigma$, and hence
\begin{equation}
  \label{eq:directional-l1-decomposition}
  \|a-b\|_1
  =
  \phi_\sigma(a)-\phi_\sigma(b).
\end{equation}
For every pair $(a,b)$, at least one of the four sign vectors satisfies
$b\preceq_\sigma a$.

\begin{proposition}[Range-tree flow reduction]
\label{prop:range-tree-flow-reduction}
A finite $\ell_1$ partial-matching instance on sets $A$ and $B$ reduces to a
directed minimum-cost flow instance having
\[
  O(m\log m+k)
\]
vertices and
\[
  O(m\log m+k\log^2 m)
\]
edges.  The minimum cost of a flow of value $k$ is equal to the minimum
partial-matching cost.

The reduction can be constructed in
\[
  O(m\log m+k\log^2 m)
\]
time.  If all coordinates are integers of absolute value at most $U$, then
all capacities are at most $k$ and all edge costs have absolute value
$O(U)$.
\end{proposition}

\begin{proof}
We first construct a routing structure for each of the four sign vectors
$\sigma\in\{-1,+1\}^2$.

\proofpart{Dominance range trees.}
For a fixed $\sigma$, transform every point $b\in B$ into
\[
  \widetilde b_\sigma
  =
  (\sigma_x b_x,\sigma_y b_y).
\]
Build a standard two-dimensional range tree on the transformed points
$\{\widetilde b_\sigma:b\in B\}$~\cite{Lueker1978}.  Thus the primary tree
is ordered by the first transformed coordinate, and every primary-tree node
stores a secondary tree ordered by the second transformed coordinate.

For every node of every secondary tree, create a routing vertex.  Direct a
zero-cost edge of capacity $k$ from each secondary-tree node to each of its
children.  A point $b\in B$ occurs as a leaf in the secondary trees of
$O(\log m)$ primary nodes.  From every such leaf occurrence, add an edge to
a single global vertex representing $b$.  Give this edge capacity $k$ and
cost
\[
  -\phi_\sigma(b).
\]
The total number of vertices and edges in all secondary trees for one sign
vector is $O(m\log m)$.

For an input point $a$ in $A$, the set
\[
  \{b\in B:b\preceq_\sigma a\}
\]
corresponds, in the transformed coordinates, to the dominance range
\[
  (-\infty,\sigma_xa_x]\times
  (-\infty,\sigma_ya_y].
\]
A standard range-tree query decomposes this range into
$O(\log^2 m)$ pairwise-disjoint canonical subsets, each represented by a
node of a secondary tree.  For every canonical node $v$ in this
decomposition, add an edge
\[
  a\longrightarrow v
\]
of capacity $1$ and cost
\[
  \phi_\sigma(a).
\]

It follows that there is a directed path from $a$ through the
$\sigma$-structure to the global vertex $b$ exactly when
$b\preceq_\sigma a$.  Every such path has cost
\[
  \phi_\sigma(a)-\phi_\sigma(b)
  =
  \|a-b\|_1
\]
by \eqref{eq:directional-l1-decomposition}.

\paragraph{The flow network.}
Use one copy of the preceding routing structure for each of the four sign
vectors.  Add a super-source $s$ and a super-sink $t$.  For each $a\in A$,
add an edge
\[
  s\longrightarrow a
\]
of capacity $1$ and cost $0$.  For each $b\in B$, add an edge
\[
  b\longrightarrow t
\]
of capacity $1$ and cost $0$.  We ask for a minimum-cost $s$--$t$ flow of
value $k$.

Given an injective matching $\mu:A\hookrightarrow B$, consider each pair
$(a,\mu(a))$.  Choose a sign vector $\sigma$ for which
$\mu(a)\preceq_\sigma a$, and route one unit of flow through the
corresponding range-tree path.  The edges into $t$ have capacity $1$, and
the matched points $\mu(a)$ are distinct, so the resulting flow is feasible.
Its cost is exactly
\[
  \sum_{a\in A}\|a-\mu(a)\|_1.
\]

Conversely, integral capacities imply the existence of an integral optimum
flow.  Decompose such a flow of value $k$ into $k$ source-to-sink paths.
Since every edge $s\to a$ has capacity $1$, and there are exactly $k$ such
edges, every point $a\in A$ starts exactly one path.  Each path enters $t$
through an edge $b\to t$, and the unit capacity of these edges ensures that
the resulting points $b$ are distinct.  The paths therefore define an
injective matching from $A$ to $B$.  Every path through a sign-$\sigma$
routing structure has cost exactly the $\ell_1$ distance between its endpoints,
so the matching cost equals the flow cost.

For each sign vector, the range-tree routing structure contains
$O(m\log m)$ vertices and edges.  The dominance queries for all points of
$A$ contribute $O(k\log^2 m)$ additional edges.  The four sign vectors
change these bounds only by a constant factor.  The source and sink edges
contribute another $O(k+m)$ edges.  This proves the stated size and
construction-time bounds.

Finally, if every coordinate has absolute value at most $U$, then
\[
  |\phi_\sigma(z)|\le 2U
\]
for every point $z$ and every sign vector $\sigma$.  Hence all nonzero
edge costs have absolute value at most $2U$, and every capacity is at most
$k$.
\end{proof}

\begin{corollary}[Algorithm for finite bipartite $\ell_1$ matching]
\label{cor:finite-range-tree-matching}
Let $N=|A|+|B|$.  If the coordinates of $A\cup B$ are integers of absolute
value $N^{O(1)}$, then finite geometric partial matching under the $\ell_1$
norm can be solved exactly in
\[
  N^{1+o(1)}
\]
time with high probability.
\end{corollary}

\begin{proof}
By \cref{prop:range-tree-flow-reduction}, the resulting flow graph has
\[
  O(N\log^2 N)
  =
  N\polylog N
\]
vertices and edges.  Its costs, capacities, and demands are polynomially bounded
integers.  Applying the almost-linear-time minimum-cost flow algorithm of
Chen et al.~\cite{ChenKyngLiuPengGutenbergSachdeva2022} gives running time
\[
  (N\polylog N)^{1+o(1)}
  =
  N^{1+o(1)}.
\]
\end{proof}

\begin{corollary}[Alternative exact algorithm for \IGR]
\label{cor:range-tree-grid-matching}
In the word-RAM model with word size $w=\Theta(\log n)$,
suppose that there is a global parameter $B=O(\log n)$
such that every input coordinate belongs to $2^{-B}\mathbb Z$
and its exact scaled representation fits in $O(1)$ machine
words.
Then \IGR under the $\ell_1$ norm can be solved exactly in
$n^{1+o(1)}$ time.
\end{corollary}
\begin{proof}
Construct the candidate set $\calC$ using
\cref{thm:general-candidate-set}.  It has size $O(n)$ and
contains the endpoints of an optimal assignment.

Let $D=2^B$.  Since every scaled input coordinate is stored in
$O(1)$ words of $\Theta(\log n)$ bits, every coordinate of $DP$
has magnitude $n^{O(1)}$.  Moreover, every $q\in\calC$ lies in
an $\ell_\infty$-square of radius at most
\[
  R_{-1}=O(\sqrt n)
\]
centered at an input point.  Hence every coordinate of $D\calC$
also has magnitude $n^{O(1)}$.

Scaling multiplies every matching cost by the common factor $D$
and preserves the optimum assignment.  We may therefore apply
\cref{cor:finite-range-tree-matching} to
\[
  \widetilde P=DP
  \qquad\text{and}\qquad
  \widetilde \calC=D\calC.
\]
Since $|\widetilde P|+|\widetilde \calC|=O(n)$, the running time is $n^{1+o(1)}$.
\end{proof}

\section{Coordinate Normalization}\label{sec:coordinate-normalization}

Choose an integer constant $K>6$, and set $\Delta=Kn$.  Partition the plane into half-open
$\Delta\times\Delta$ cells, and connect two occupied cells whenever their
cell indices differ by at most one in each coordinate.  Let
$P_1,\ldots,P_s$ be the sets of input points in the resulting connected
components.  Any two distinct components are at $\ell_\infty$-distance
greater than $\Delta$.  On the other hand, for a component $P_j$ of size
$n_j$, \cref{lem:global-safe-radius} guarantees an optimal assignment in
which every point moves by at most
$3\sqrt{n_j}\le3\sqrt n<\Delta/2$.  Consequently, optimal assignments
chosen independently for different components use disjoint grid points.
Since the restriction of any feasible global assignment to $P_j$ is a
feasible assignment for $P_j$, the optimum of the original instance is
exactly the sum of the component optima.

Each component contains at most $n$ occupied cells, and hence has coordinate
diameter $O(n\Delta)=O(n^2)$.  We may therefore translate each component
by an integer vector, without changing any assignment cost or feasibility
condition, so that all of its integer coordinate parts have
$O(\log n)$ bits.  Together with the assumed $O(\log n)$ fractional
precision, every normalized coordinate is representable in $O(1)$ machine
words.  The components of occupied cells and their translations can be found
by hashing the occupied cells and traversing their constant-degree
adjacency graph, in expected linear time in the input representation.  We henceforth
suppress this decomposition and normalization and assume that the input
coordinates are already normalized.

\section{Target-Restricted Variants}\label{sec:endpoint-restrictions}

We briefly record two variants that often arise in real-world applications.

\begin{proposition}[Finitely many unavailable grid points]\label{prop:unavailable-grid-points}
Let $F\subseteq\Z^2$ be a finite set of $m$ unavailable grid points.  The variant in which the assignment must avoid $F$ reduces to the unrestricted problem on $n+m$ input points.  Consequently, under the same assumption as in \cref{thm:main-algorithm}, for the $\ell_1$ norm, there is a randomized algorithm that outputs an optimal assignment in expected $\widetilde O(n+m)$ time.
\end{proposition}
\begin{proof}
For every unavailable grid point $b\in F$, add a dummy input point $d_b$ located at $b$, and solve the unrestricted instance on
\[
  P'=P\uplus\{d_b:b\in F\}.
\]
Any feasible assignment of the original points to $\Z^2\setminus F$ extends to an assignment of $P'$ with the same cost by sending each dummy point $d_b$ to the grid point $b$.

Conversely, take any injective assignment of $P'$.  We show that it can be transformed, without increasing its cost, into one that sends every dummy point $d_b$ to $b$.  Process the dummy points one at a time.  If $d_b$ is assigned to some grid point $q\ne b$, move it to $b$.  If some point $x$ was previously assigned to $b$, reassign $x$ to $q$.  This preserves injectivity and does not disturb any dummy point that has already been fixed.  The cost of the dummy point decreases by $\normone{b-q}$, while the cost of $x$ increases by at most $\normone{b-q}$ by the triangle inequality.  Thus the total cost does not increase.  Repeating this operation fixes all dummy points.  After the dummy points occupy the unavailable grid points, the original points are assigned injectively to $\Z^2\setminus F$.

The optimal values of the restricted and augmented unrestricted instances are therefore equal.  Applying \cref{thm:main-algorithm} to the augmented instance gives the claimed running time.
\end{proof}

\begin{proposition}[Thick rectangular window]\label{prop:rectangular-window}
There is a constant $K>0$ such that the following holds.  Let
\[
  R=[\ell_x,u_x]\times[\ell_y,u_y]
\]
be an axis-aligned rectangle with integer endpoints whose width and height are both at least $K\sqrt n$, and let $S_R=R\cap\Z^2$. The input points themselves need not lie in $R$.  For the $\ell_1$ norm, the variant in which all grid points used by the assignment must lie in $S_R$ can be solved exactly in expected $\widetilde O(n)$ time under the same assumption as in \cref{thm:main-algorithm}.
\end{proposition}

\begin{proof}
Let $\bar h_i$ be the coordinatewise projection of $h_i$ onto $R$.  If $h_i\in R$, then $\bar h_i=h_i$; otherwise, $\bar h_i$ lies on the boundary of $R$, possibly at a corner.  For every $q\in S_R$,
\[
  \normone{h_i-q}
  =
  \normone{h_i-\bar h_i}+\normone{\bar h_i-q}.
\]
The first term is independent of $q$, so it can be added to the objective after solving the assignment problem for the projected points $\bar h_i$.

We next run the pruning algorithm on the projected points but restrict
candidate endpoints to $S_R$.  Use the radii
$R_{-1},R_0,\ldots,R_T$ from \cref{sec:radii-seq}.  Every such radius
$r$ is a power of two and satisfies
\[
  2\le r\le R_{-1}<6\sqrt n.
\]
We claim that, for large enough $K$, for every $h\in R$ and every one of these radii,
\[
  |\BallOne(h,r)\cap S_R|\ge \frac{r^2}{4}.
\]
At $r=R_{-1}$, this lower bound is at least $9n/4>n$.  Thus no optimal
assignment to $S_R$ can send an input farther than $R_{-1}$: otherwise
every grid point in this ball would have to be occupied, which is
impossible with only $n$ inputs.  Hence the initial radius vector is safe.  At every
test radius $R_t$, the same $r^2/4$ lower bound makes the safe-pruning
proofs in \cref{lem:safe-pruning,lem:simultaneous-safe-pruning} valid with
$S_R$ in place of $\Z^2$.  The high-radius charging proof then applies
unchanged.  The final candidate set $\calC^R$ is obtained by intersecting
the usual candidate set with $S_R$.  It has size $O(n)$ and contains the
endpoints of an optimal assignment of the projected points to $S_R$.

We construct the graph from \cref{sec:l1-flow} for the projected points $\bar h_i$, using $\calC^R$ as the candidate set, solve the resulting instance with the minimum-cost flow algorithm of Dong et al.~\cite{DongGaoGoranciLeePengSachdevaYe2025}, and add $\sum_{i=1}^n\normone{h_i-\bar h_i}$ to its optimal value.
\end{proof}

\end{document}